\documentclass[fleqn,11pt]{article}
\usepackage{amssymb,latexsym}
\usepackage{amsmath, amsthm}
\usepackage[cp1250]{inputenc}   
\usepackage{color}
\DeclareFontFamily{OT1}{rsfs}{} \DeclareFontShape{OT1}{rsfs}{m}{n}{
<-7> rsfs5 <7-10> rsfs7 <10-> rsfs10}{}
\DeclareMathAlphabet{\mycal}{OT1}{rsfs}{m}{n}
\newtheorem{Theorem}{Theorem}
\newtheorem*{Hypothesis}{Hypothesis}
\newcommand{\rd}{{\rm d}} 

\begin{document}
\title{\bf Towards uniqueness of degenerate axially symmetric Killing horizon}

\author{Jacek Jezierski\thanks{E-mail: Jacek.Jezierski@fuw.edu.pl} \ and
Bartek Kami\'nski\thanks{E-mail: fizyk20@gmail.com}\\
Department of Mathematical Methods in Physics, \\
Faculty of Physics, University of Warsaw,\\
ul. Ho\.za 74, 00-682 Warsaw, Poland}
\maketitle

\begin{abstract}
For near horizon geometry we examine the linearized equations around extremal Kerr horizon (which is a unique axially symmetric near horizon geometry) and give some arguments towards
stability of this horizon with respect to generic (non-symmetric) linear perturbation
of near horizon geometry. The result is also applicable for other situations like Kundt's class spacetimes or
isolated horizons.
\end{abstract}

\section{Introduction}
Let us consider
the following {\em basic equation} on a two-dimensional compact manifold
\begin{equation}
\label{basic}
  \omega_{A||B} +\omega_{B||A} +2\omega_A \omega_B = R_{AB} \, ,
\end{equation}
where $\omega_A \rd x^A$ is a covector field, $||$ denotes
covariant derivative compatible with the metric $g_{AB}$, and
$R_{AB}$ is its Ricci tensor. The equation (\ref{basic})
is a starting point of our considerations and it is a special case of
(3.7) in \cite{ABL}, if we assume that ${\tilde S}_{AB}$ vanishes.
See also \cite{CRT} or \cite{Lucietti}.

 Some geometric consequences of the
 {\it basic equation}\footnote{This is an equation describing the so-called near horizon geometries
 see \cite{nearhorizon}.} (resulting from Einstein equations)
 were discussed in \cite{JJcqg26}.
This is a non-linear PDE for unknown covector field and unknown
Riemannian structure on the two-dimensional manifold.
It appears in the context of
Kundt's class metrics (cf. \cite{Kundt}), degenerate Killing horizons
\cite{CRT}, \cite{Lucietti}, or vacuum degenerate isolated horizons \cite{ABL},
\cite{LP}, \cite{LPJJ}.
Several important results are already proved, like
topological rigidity of the horizon and integrability conditions
(cf. \cite{JJcqg26}).
Moreover, when the
one-form $\omega_B\rd x^B$ is closed (e.g. static degenerate horizon
\cite{CRT}) there are no solutions of (\ref{basic}).
The transformation (\ref{omPhi}) of a covector $\omega_A$ leads
to (partially) linear problem (invented in \cite{JJcqg26}) and simplifies the proof of
the uniqueness of extremal Kerr for axially symmetric horizon.
However, the problem of the existence of non-symmetric solutions
to the {\it basic equation} remains open. The solutions of this equation enables one to construct near horizon metric
(cf. \cite{CRT}, \cite{Lucietti}, \cite{nearhorizon}), Kundt's class spacetime or isolated horizon neighborhood.

In \cite{JJcqg26} the following results were proved:
\begin{Theorem}\label{main}
For any Riemannian metric $g_{AB}$
on a two-dimensional, compact, connected manifold $B$
with no boundary and genus $\mathbf{g}\geq 2$ there are
no solutions of basic equation.
\end{Theorem}
\begin{Theorem}
\label{torus}
For any Riemannian metric $g_{AB}$
on a two-dimensional torus
 equation (\ref{basic}) possesses only trivial
solutions $\omega^A\equiv 0 \equiv K$ and the metric $g_{AB}$
is flat.
\end{Theorem}
\begin{Theorem}\label{tw3}
There are no solutions of equation (\ref{basic})
with the following properties:
\begin{itemize}
\item $\omega^A=0$ only at finite set of points,
\item $B$ is a sphere with non-negative Gaussian curvature.
\end{itemize}
\end{Theorem}
The symmetric part of $\omega_{A||B}$ is controlled by the equation but
$f:=\frac12\omega_{A||B}\varepsilon^{AB}$ is an unknown function on a sphere.
We have
\begin{equation}\label{basic2}
\omega_{A||B}= f\varepsilon_{AB}+\frac12 K g_{AB} -\omega_A\omega_B \, .
\end{equation}
{\color{red} The integrability condition:
\begin{equation}\label{laplasjanR}
\frac14 R^{||A}{_A} +\frac32(R\omega^A)_{||A}=
 6f^2+ \frac38 R ( R-8\omega_A\omega^A )
\end{equation}
 implies that there exists
non-empty open subset, where $8\omega_A\omega^A > R > 0$.}

In this paper we
analyze a linear perturbation of extremal Kerr solution.
More precisely, in Section 2 we perform linearization of eq. (\ref{basic}) around extremal Kerr solution (\ref{gKerr}).
Axial symmetry of the background solution gives possibility to decompose linearized solution into Fourier series.
Each Fourier mode $v_k$ fulfills ordinary differential equation (\ref{eqn:final}). Using functional analysis methods we prove (in Appendix C) that there are no regular solutions for $|k|>8$. We hope to check numerically the nonexistence of low modes for $|k|\leq 8$.

Moreover, in this Section below we give some new results like equivalent formulations of the full nonlinear problem (cf.
Theorem \ref{t4}), equivalence between (\ref{basic2}) and (\ref{basicPhi}) or some properties contained in formulae
at the end of Section 1.

Finally, some nontrivial calculations are shifted to the Appendix which also contains some useful formulae.

\subsection{Transformation to linear problem}

Let us denote
\begin{equation}\label{omPhi} \Phi_A: = \frac{\omega_A}{\omega^B\omega_B} \,. \end{equation}
For any domain, where $\omega^B\omega_B > 0$,
equation (\ref{basic}) implies
\begin{equation}\label{rotPhi} \Phi_{A||C}\varepsilon^{AC}
= \left(\frac{\omega_A}{\omega^B\omega_B}\right)_{||C}
\varepsilon^{AC} = 0 \end{equation}
which simply means that the one-form $\Phi_A\rd x^A$ is closed,
and locally there exists a coordinate $\Phi$ such that
\[ \rd\Phi= \Phi_A\rd x^A \, . \]
Moreover, from (\ref{basic}) we get
\begin{equation}\label{divPhi} \Phi^A{_{||A}} = 1 \end{equation}
hence the potential $\Phi$ is a solution of the Poisson's equation:
\begin{equation}\label{lapPhi}
 \triangle\Phi = 1 \, .
\end{equation}
\underline{Remark} If we choose one isolated point, where $\omega$
vanishes, then for a given metric $g$ we have unique solution of the
above Laplace-Beltrami equation (Green's function in the enlarged
sense cf. \cite{RS-JMP47}). For more isolated points we can take linear combination of such
solutions. More precisely, let $G_{x_0}$ be a unique solution (for a given
metric $g$) of the equation (\ref{lapPhi}) on $S^2-\{x_0\}$.
If $c_0+ c_1+ \ldots + c_n = 1$ (where $c_i \in {\mathbb R}$)
then $\Phi= c_0 G_{x_0}+ c_1 G_{x_1}+ \ldots + c_n G_{x_n}$ is
a solution of (\ref{lapPhi}) on $S^2-\{ x_0, x_1, \ldots ,x_n \}$, and
$\omega$ vanishes at the points $x_0, x_1, \ldots, x_n$.

\subsection{Two zeros of $\omega$}
Suppose $\omega_A$ vanishes at two distinct points in a generic way
(i.e. $\omega_{A||B}$ is non-degenerate at those points).
Then the equations (\ref{divPhi}) and (\ref{rotPhi}) extend (in the sense of distributions) as follows:
\begin{equation}\label{divPhi1} \Phi^A{_{||A}} = 1 -c_1{\boldsymbol\delta}_{\theta=\pi}-c_2 {\boldsymbol\delta}_{\theta=0}\end{equation}
\begin{equation}\label{rotPhi1} \Phi_{A||C}\varepsilon^{AC}= d_1{\boldsymbol\delta}_{\theta=\pi} -
d_2 {\boldsymbol\delta}_{\theta=0}
\end{equation}
Integration of the above equations on $S^2$ implies $d_1=d_2=d$ and $c_1+c_2=\int \lambda=$(total volume of $S^2$).
Hence, for $\Phi_A=\partial_A\Phi+\varepsilon_A{^B}\partial_B\tilde\Phi$ the potentials $\Phi, \tilde\Phi$
fulfill Laplace equations:
\begin{equation}\label{LaPhi1} \triangle \Phi = 1 -c_1{\boldsymbol\delta}_{\theta=\pi}-c_2 {\boldsymbol\delta}_{\theta=0}\end{equation}
\begin{equation}\label{LaPhi2} \triangle \tilde\Phi = d_1{\boldsymbol\delta}_{\theta=\pi} -
d_2 {\boldsymbol\delta}_{\theta=0}\end{equation}
and their solutions may be
expressed in terms of generalized Green's functions on $S^2$ which are well defined as the distributions
(they are integrable functions, smooth outside poles with log divergence at poles).

Moreover, the trace of (\ref{basic})
\begin{equation}
\label{divom1}
  \omega^A{_{||A}} = K - \omega^A \omega_A
\end{equation}
may be expressed
in terms of $\Phi^A$ as follows:
\begin{equation}\label{PhiK1}
\partial_A\left(\frac{\lambda\Phi^A}{\Phi^B\Phi_B}\right)+
\frac{\lambda}{\Phi^B\Phi_B}-\lambda K =0 \quad \equiv \quad \frac2{\|\Phi\|}\partial_A\left(\frac{\lambda\Phi^A}{\|\Phi\|}\right) 
=\lambda K \, .
\end{equation}

\subsection{Equivalent form of the basic equation in terms of the covector $\Phi_A$ and its conformal rescaling}
Equations (\ref{rotPhi}) and (\ref{divPhi}) together with (\ref{PhiK1}) written as follows:
\begin{eqnarray}\label{rotPhi2} \lambda\varepsilon^{AC}\partial_C\Phi_{A}
&=& 0\, , 
\\\label{fil} \partial_A\left( \lambda g^{AB}\Phi_B \right) &=& \lambda\, , \\
\label{PhiK2}
\partial_A\left(\frac{\lambda g^{AB}\Phi_B}{g^{CD}\Phi_C\Phi_D}\right)+
\frac{\lambda}{g^{CD}\Phi_C\Phi_D}-\lambda K &=& 0 \, ,
\end{eqnarray}
for the conformally equivalent metric $h_{AB}=\exp(-2u)g_{AB}$ (cf. eq. (\ref{gKerr1})) are almost the same
\[ (\lambda K)(h)-(\lambda K)(g)= \lambda_h \triangle_h u = \lambda_g \triangle_g u \, , \quad -(\lambda K)(h)= \frac12 a^2{_{,xx}} \, ,\]
\[ \partial_A\left( \lambda_h h^{AB}\Phi_B \right)=\lambda_h \exp(2u) \, , \quad \lambda\varepsilon^{AC}\partial_C\Phi_{A}=0 \, , \]
\[ \partial_A\left(\frac{\lambda_h h^{AB}\Phi_B \exp(2u)}{h^{CD}\Phi_C\Phi_D}\right)+
\frac{\lambda \exp(4u)}{h^{CD}\Phi_C\Phi_D}-\lambda_h K_h +\lambda_h \triangle_h u =0 \, .\]
Moreover, we have the following
\begin{Theorem}\label{t4}
Equations (\ref{rotPhi2}--\ref{PhiK2}) are locally equivalent to the eq. (\ref{basic2}) in the domain, where $\omega_A=
\frac{\Phi_A}{\Phi^B\Phi_B}$ is not vanishing.
\end{Theorem}
\begin{proof}
Let us represent tensor $\omega_{A||B}$ as a sum of three parts: skewsymmetric ($f$), traceless symmetric ($\tau_{AB}$)
and trace ($\tau$):
\begin{equation}\label{deromega}
\omega_{A||B}= f\varepsilon_{AB}+ \tau_{AB} + \tau g_{AB} \, .
\end{equation}
We have to show that $\tau_{AB}$ and $\tau$ are determined by eq. (\ref{rotPhi2}--\ref{PhiK2}).
It is easy to check that (\ref{PhiK2}) implies $2\tau=K-\|\omega\|^2=\omega^A{_{||A}}$.
Moreover, (\ref{rotPhi2}) gives
\[ \varepsilon^{AB}\omega_A\omega^C\tau_{CB}=0 \]
and similarly (\ref{fil}) implies
\[ 2\omega^A\omega^B\tau_{AB}=-\|\omega\|^4 \, .\]
Let us observe that any two-dimensional traceless symmetric tensor has only two independent components, hence
the last two conditions determine $\tau_{AB}$ uniquely in the following form:
\[ \tau_{AB}=-\omega_A\omega_B +\frac12 g_{AB} \|\omega\|^2 \, .\]
Finally, the above formula together with
$\tau=\frac12 (K-\|\omega\|^2)$ give the eq. (\ref{basic2}).
\end{proof}
One can also check the following formula:
\begin{equation}\label{basicPhi} \Phi_{A||B}= 0\cdot\varepsilon_{AB} +\frac12 g_{AB} - f (\ast\Phi_A\Phi_B+\ast\Phi_B\Phi_A) +
(1-K\|\Phi\|^2)\left(\frac{\Phi_A\Phi_B}{\|\Phi\|^2} -\frac12 g_{AB}\right)
\end{equation}
which is equivalent to (\ref{basic2}) but in terms of $\Phi$.

Let us observe that $\Phi^B{_{||BA}}=0$ hence
the symmetry of the tensor $\Phi_{A||B}$ implies
\[ \Phi_A{^{||B}}_B=\Phi^B{_{||AB}}=\Phi^B{_{||AB}}-\Phi^B{_{||BA}}=R_{AB}\Phi^B \, ,
\]
and we obtain the following nice formulae:
\begin{equation}\label{lplphi}
    \Phi^{A||B}{_B}=K \Phi^A \, , \quad \ast\Phi^{A||B}{_B}=K {\ast\Phi}^A
    \, .
\end{equation}
Moreover,
\[ \Phi^{A||B}\omega_{A||B}= K-\|\omega\|^2 \]
and 
\[ (\Phi^{A||B}\omega_A )_{||B} = \Phi_A{^{||B}}_B\omega^A +  \Phi^{A||B}\omega_{A||B}= 2K-\|\omega\|^2 \, ,\]
\[ \lim_{\epsilon\rightarrow 0^+} \int_{\partial S_\epsilon}\Phi^{A||B}\omega_A \rd S_B =\int_{S^2} K =4\pi \quad
\mbox{where} \; S_\epsilon := S^2 \setminus \left(\bigcup_{x_i\in \, \omega^{-1}(\{0\})} K(x_i, \epsilon)\right) \, . \]

\section{Linearization of basic equation around extremal Kerr}
After introducing a new coordinate $x:=\cos\theta$
the (two-dimensional) extremal Kerr (see \cite{JJcqg26})
\begin{equation}\label{gKerr}
g_{\mbox{\tiny\rm Kerr}}=2m^2 \left[ \frac{1+\cos^2\theta}2\rd\theta^2
+ \frac{2\sin^2\theta}{1+\cos^2\theta}\rd\phi^2 \right] \, ,\end{equation}
\begin{equation}\label{omKerr}
 \omega^\theta=-\frac{\sin\theta\cos\theta}{m^2(1+\cos^2\theta)^2} \, ,\quad
   \omega^\varphi=\frac{1}{2m^2(1+\cos^2\theta)} \, ,
\end{equation}
takes the following form:
\begin{equation}\label{gKerr1}
g_{\mbox{\tiny\rm Kerr}}=h_{AB}\rd x^A\rd x^B
=2m^2\left(a^{-2} \rd x^2 +a^{2}\rd\varphi^2\right)
\, , \end{equation}
 where $\displaystyle a^2:= 2\frac{1-x^2}{1+x^2}$ and
$\displaystyle\lambda:=\sqrt{\det h_{AB}}=2m^2$.
 The components of various objects for Kerr are the following:
\[ \omega_x=\frac{x}{1+x^2} \, ,\quad
   \omega_\varphi=\frac{a^2}{1+x^2} \, , \quad
 \|\omega\|^2=\frac1{2m^2}\frac{a^2}{1+x^2} \, .\]
\begin{equation}\label{phi1}
    \frac1{2m^2}\Phi=\frac{x}{a^2}\rd x +\rd\varphi \; , \quad
    \frac1{2m^2}\ast\Phi=\frac{1}{a^2}\rd x -x\rd\varphi \quad \left(\ast\Phi_A:=\varepsilon_A{^B}\Phi_B\right)\, ,
\end{equation}
\begin{equation}
 K=\frac2{m^2}\frac{1-3x^2}{(1+x^2)^3} \; , \quad
    f=\frac1{m^2}\frac{x(1+a^2)}{(1+x^2)^2}=\frac1{m^2}\frac{x(3-x^2)}{(1+x^2)^3}\; ,
    \quad \frac{K}2 +if=\frac1{m^2(1-ix)^3} \, ,
\end{equation}
\begin{equation}\label{phi2}
    \|\Phi\| = \| \ast\Phi \| \; , \quad
    \ast\Phi \wedge \Phi = \|\Phi\|^2\lambda\rd x\wedge\rd\varphi  \, .
\end{equation}

The nearby metric $g$ we describe by conformal factor:
\begin{equation}\label{defu}
    g_{AB}=\exp(2u)h_{AB}
\end{equation}
and we get
\[ \Gamma^C{_{AB}}(g)=\Gamma^C{_{AB}}(h)+S^C{_{AB}} \, , \]
\[ S^C{_{AB}}=\delta^C_A\partial_B u + \delta^C_B\partial_A u
    - h_{AB}h^{CD}\partial_D u \, .\]
Let us denote by $u^B:= h^{BA}\partial_A u$ the gradient of $u$ with respect to the metric $h$.
We have
\begin{eqnarray}\label{omegaab}
 \nabla_B(g)\omega_A &=& \nabla_B(h)\omega_A-S^C{_{AB}}(u)\omega_C\\
 & =& \nabla_B(h)\omega_A + h_{AB}\omega_C u^C -\omega_A u_B -\omega_B u_A
 \end{eqnarray}
Moreover, the Gaussian curvatures $K_h$ and $K_g$ for
 the conformally related metrics $h$ and $g$ respectively are related as follows
\[ \triangle_h u = K_h-\exp(2u)K_g \, .\]
This gives the following transformation for the right-hand side of (\ref{basic}):
\begin{equation}\label{RAB} R_{AB}(g)=K_g g_{AB}=\left(K_h- \triangle_h u\right)h_{AB} \, .\end{equation}
Using (\ref{omegaab}) and (\ref{RAB}) we rewrite basic equation (\ref{basic}) as follows:
\begin{equation}
\nabla_B(h)\omega_A +\nabla_A(h)\omega_B +
2\left(h_{AB}\omega_C u^C - \omega_A u_B - \omega_B u_A + \omega_A\omega_B\right)=
\left(K_h- \triangle_h u\right)h_{AB} \, .
\end{equation}
Let us denote the linear part of the covector $\omega$ by
\[ {\mathsf w}_A := \omega_A - \omega^{\mbox{\tiny\rm Kerr}}_A \, .\]
 Now we are ready to linearize basic equation.
\begin{eqnarray} \nonumber
& & 2\left( \omega^{\mbox{\tiny\rm Kerr}}_A {\mathsf w}_B + \omega^{\mbox{\tiny\rm Kerr}}_B {\mathsf w}_A
+h_{AB}\omega^{\mbox{\tiny\rm Kerr}}_C u^C -\omega^{\mbox{\tiny\rm Kerr}}_A u_B
-\omega^{\mbox{\tiny\rm Kerr}}_B u_A \right) + \nabla_B(h){\mathsf w}_A +\nabla_A(h){\mathsf w}_B  \\ & & + h_{AB}\triangle_h u  =
2{\mathsf w}_A u_B+2{\mathsf w}_B u_A - 2h_{AB}{\mathsf w}_C u^C -2 {\mathsf w}_A {\mathsf w}_B \approx 0
\end{eqnarray}

Finally,
for covector ${\mathsf w}_A$ and conformal factor $u$
in (\ref{defu})
the linearization of (\ref{basic}) takes the following form:
\begin{equation}\label{divw}
\nabla_A ({\mathsf w}^A+u^A)+2\omega^A {\mathsf w}_A=0 \, , 
\end{equation}
\begin{equation}\label{TSw}
TS\left(\nabla_A {\mathsf w}_B + 2\omega_A({\mathsf w}_B-u_B)\right)=0 \, ,
\end{equation}
where now $\omega$ and $\nabla$ are background objects
(corresponding to the Kerr solution (\ref{gKerr1})), and
\[ TS(t_{AB}):= t_{AB}+t_{BA}-h_{AB}h^{CD}t_{CD} \]
denotes the traceless symmetric part of the tensor $t_{AB}$.

We show in Appendix \ref{redukcja} that after elimination of $u_A$ we get:
\begin{equation}\label{rotw}
\triangle_h ({\mathsf w}_A{\ast\Phi}^A)+\varepsilon^{AB} {\mathsf w}_{A||B}=0 \, ,
\end{equation}
\begin{equation}\label{laplw}
\triangle_h ({\mathsf w}_A\Phi^A)+ 4 {\mathsf w}_A\Phi^A \|\omega\|^2 +3 {\mathsf w}^A{_{||A}}=0 \, ,
\end{equation}
where
\begin{equation}
\label{eqn:u}
u_A = \frac{1}{2}\left[{\mathsf w}_A + \nabla_B(\Phi^B {\mathsf w}_A - \Phi_A {\mathsf w}^B) + \nabla_A(\Phi^B {\mathsf w}_B) \right] \, .
\end{equation}

\underline{Remark:} The equations (\ref{rotw}--\ref{laplw}) are
{\em conformally} covariant with respect to the rescaling of the two-metric $h$.
More precisely, the form of these equations is the same for two conformally
related metrics provided that $\Phi$, $\ast\Phi$ are vector fields
and $\mathsf w$ and $\omega$ are covector fields.
One can easily verify this observation multiplying the above equations
by scalar density $\lambda$.

The non-existence of the solution ${\mathsf w}_A$ to the equations
(\ref{rotw}--\ref{laplw}) is equivalent to the
stability of the solution (\ref{gKerr}--\ref{omKerr}).

Axial symmetry of the background solution
enables one to separate variable $\varphi$
with the help of Fourier transform and (\ref{rotw}-\ref{laplw}) becomes
second order ODE for \[ {\mathsf w}: [-1,1] \mapsto {\mathbb R}^2 \, . \]

\noindent One can also introduce another pair of variables: 
\begin{equation}\nonumber
\alpha :=  \varepsilon^{AB} {\mathsf w}_A \Phi_B = \frac{1}{2} \left[ 2{\mathsf w}_x - \frac{x(1+x^2)}{1-x^2}{\mathsf w}_\varphi \right] =
  {\mathsf w}_x - \frac{x(1+x^2)}{2(1-x^2)} {\mathsf w}_\varphi = {\mathsf w}_x -\frac{x}{a^2}{\mathsf w}_\varphi \, ,
\end{equation}
\begin{equation}\nonumber
\beta :=  \Phi_A {\mathsf w}^A = m^2 \left[ \frac{x(1+x^2)}{1-x^2}{\mathsf w}^x + 2{\mathsf w}^\varphi \right]
=  x {\mathsf w}_x + \frac{1+x^2}{2(1-x^2)}{\mathsf w}_\varphi = x{\mathsf w}_x +\frac{1}{a^2}{\mathsf w}_\varphi \, ,
\end{equation}
\noindent where $\displaystyle a^2 := 2\frac{1-x^2}{1+x^2}$. The formula (\ref{eqn:u}) takes a simple form:
\begin{equation}
\label{uab}
u_A = \frac{1}{2}\left[{\mathsf w}_A + \varepsilon_A{^B} \nabla_B \alpha + \nabla_A \beta \right] \, .
\end{equation}

Moreover, the inverse transformation
\begin{equation}\label{newabis}
     {\mathsf w}_x =\frac{\alpha+x\beta}{1+x^2} \, , \quad
    {\mathsf w}_\varphi = a^2\frac{\beta-x\alpha}{1+x^2} \, ,
\end{equation}
implies the following form of the equations (\ref{rotw}--\ref{laplw})
in terms of variables $\alpha$, $\beta$:
\begin{equation}\label{rotwab}
\triangle_h (\alpha)+\partial_\varphi\left(\frac{\alpha+x\beta}{1+x^2}\right)
-\partial_x\left(a^2\frac{\beta-x\alpha}{1+x^2}\right)=0 \, ,
\end{equation}
\begin{equation}\label{divwab}
\triangle_h (\beta)+ \frac{4 a^2}{1+x^2}\beta +
3 \partial_\varphi\left(\frac{\beta-x\alpha}{1+x^2}\right)+
3\partial_x\left(a^2\frac{\alpha+x\beta}{1+x^2}\right)
=0 \, ,
\end{equation}
where
\[ \triangle_h := \partial_x a^2 \partial_x +
\partial_\varphi a^{-2} \partial_\varphi \, .\]

Let us denote $v:=\left[\begin{array}{c} \alpha \\ \beta \end{array} \right]$,
$\displaystyle B:=\frac1{1+x^2}\begin{bmatrix}
x & -1 \\ 3 & 3x
\end{bmatrix}$, $\displaystyle C:=\frac1{1+x^2}\begin{bmatrix}
1 & x \\ -3x & 3
\end{bmatrix}$, then the equations (\ref{rotwab}--\ref{divwab})
take the following (matrix) form:
\begin{equation}\label{matrixform}
\triangle_h v + \frac{4 a^2}{1+x^2} {\begin{bmatrix}
0 & 0 \\ 0 & 1\end{bmatrix}}v +
\partial_x\left(a^2 B v\right) +
\partial_\varphi\left(C v\right)=0 \, .
\end{equation}
Other useful identities:
\[ {\mathsf w}_A=\beta\omega_A+\alpha\varepsilon_A{^B}\omega_B=\partial_A (2u-\beta)-\varepsilon_A{^B}\partial_B\alpha \, ,\]
\[ 2\partial_A u=\partial_A\beta +\beta\omega_A +\varepsilon_A{^B}(\partial_B\alpha+\alpha\omega_B) \, . \]

\subsection{Boundary data}
A small perturbation of Kerr data (\ref{gKerr}--\ref{omKerr}) does not destroy the number of two zeros for covector $\omega_A$.
This is a simple consequence of the ``inverse function theorem''. More precisely, the non-vanishing curvature
in the neighborhood of ``spherical pole'' (zero of $\omega_A$) assures invertibility of the first derivative
$ \nabla_A \omega_B$ in a small open neighborhood\footnote{Formula (\ref{basic2}) implies that
$\det  \nabla_A \omega^B=f^2+\frac K2(\frac K2 - \|\omega\|^2)$, for Kerr
$\det  \nabla_A \omega^B=\frac{2x^2(3-x^2)}{(1+x^2)^5}$ and it vanishes only on the equator $x=0$.} and implies existence of a local diffeomorphism
$\omega_A(x^B)$. Hence, for perturbed $\omega_A(x^B)$ there exists
(in a small open neighborhood of spherical pole) precisely one point, where $\omega_A$ vanishes.
The freedom of global conformal transformations enables one to introduce ``new conformal coordinates''
in such a way that the spherical poles are always at the points where $\omega_A$ vanishes.
Hence, we can always assume that the perturbed $\omega_A$ vanishes at spherical poles which implies
zero (homogeneous) boundary data for linear perturbation ${\mathsf w}_A$ or equivalently for $v=(\alpha, \beta)$\footnote{It is not obvious that ${\mathsf w}_A=0$ corresponds to $v=0$ and it is not true for $k=0$.}.

One can also show that respectively chosen conformal vector field $X$ enables one to change ${\mathsf w}_A \rightarrow {\mathsf w}_A +{\cal L}_X\omega_A$ in such a way that it will vanish at a given point (see appendix \ref{cvf}).

\begin{Hypothesis}\label{stability}
The equation (\ref{matrixform}) has no regular solutions for homogeneous boundary data ${{\mathsf w}_A}{\big|_{x=1}}=0={{\mathsf w}_A}{\big|_{x=-1}}$.
\end{Hypothesis}
\noindent
{\em Proof attempt}.
Let us consider Fourier series for $v$:
\[v(x,\phi) = \sum_{k=-\infty}^{\infty} v_k(x) e^{ik\phi}\]
It leads to ODE for $v_k(x)$:
\begin{equation}\label{eqn:final}
\partial_x a^2 \partial_x v_k - \frac{k^2}{a^2} v_k + \frac{4a^2}{1+x^2} \left[ \begin{array}{cc}
0 & 0 \\
0 & 1
\end{array} \right] v_k + \partial_x (a^2Bv_k) + ik(Cv_k) = 0
\end{equation}
We check that $v_k$ vanishes at poles for $|k|\geq 1$, because ${\mathsf w}_A$ vanishes there.
For $k=0$ we have axial symmetry, hence we already have uniqueness in full nonlinear case,
however it would be nice to check this fact independently.


For $|k|>8$ we show in appendix C that there are no regular solutions.
There are some initial numerical results which confirm nonexistence hypothesis for $|k|\leq 8$.
We are going to check numerically the existence or nonexistence of low modes. 
The results will be published in a separate paper.
\hfill $\Box$

The above Hypothesis implies stability of the extremal Kerr horizon.
It is true for $|k|> 8$ and adding this assumption we get Theorem \ref{mainth}. More precisely, eq. (\ref{eqn:final}) has no regular solutions for $|k|> 8$.
The analytical proof is given in Appendix C.

\appendix
\section{Kerr in conformal coordinates}
The background metric (\ref{gKerr1}) can be conformally related to
unit sphere metric as follows:
\begin{equation}\label{gKerrc1}
g_{\mbox{\tiny\rm Kerr}}=h_{AB}\rd x^A\rd x^B
=2m^2\left(a^{-2} \rd x^2 +a^{2}\rd\varphi^2\right)
= 2m^2 F^2 {\tilde h}_{AB}\rd {\tilde x}^A\rd {\tilde x}^B
\end{equation} where
\[ {\tilde h}_{AB}\rd {\tilde x}^A\rd {\tilde x}^B :=
\left[ \frac{\rd{\tilde x}^2}{1-{\tilde x}^{2}}
+(1-{\tilde x}^{2})\rd\tilde\varphi^2\right]
 \, ,\quad {\tilde\varphi=\varphi}
  \, , \quad {\tilde x}= \frac{x-\tanh\frac{x}2}{1-x\tanh\frac{x}2}\]
 is the usual unit sphere metric and
\[ F^2=\frac{a^2}{1-{\tilde x}^{2}}
=\frac2{1+x^2}\left(\cosh\frac{x}2-x\sinh\frac{x}2\right)^2 \, ,
\quad \rd x=F^2\rd{\tilde x} \, .\]

\section{Reduced linearized equations}\label{redukcja}

\subsection{Elimination of $u_A$}\label{sec:elim}

We start from traceless part (\ref{TSw}):
\begin{equation}
\begin{split}
\label{eqn:traceless2}
& \nabla_A{\mathsf w}_B + \nabla_B{\mathsf w}_A + 2\omega_A{\mathsf w}_B + 2\omega_B{\mathsf w}_A - \nabla^C{\mathsf w}_Ch_{AB} - 2\omega^C{\mathsf w}_Ch_{AB} + \\
& - 2\omega_Au_B - 2\omega_Bu_A + 2\omega^Cu_Ch_{AB} = 0\, .
\end{split}
\end{equation}
The two independent components $(AB) = (xx)$ and $(AB) = (x\phi)$ can be written as follows.
 Component $(xx)$:
\begin{equation}
\begin{split}\nonumber
& 2\nabla_x {\mathsf w}_x + 4\omega_x {\mathsf w}_x - (\nabla^x{\mathsf w}_x + \nabla^\phi {\mathsf w}_\phi)h_{xx} - 2(\omega^x{\mathsf w}_x + \omega^\phi {\mathsf w}_\phi)h_{xx} +\\
& - 4\omega_x u_x + 2 (\omega^xu_x + \omega^\phi u_\phi) h_{xx} = 0
\end{split}
\end{equation}
or in an equivalent form (dividing by $h_{xx}$):
\begin{equation}
\begin{split}\nonumber
& 2\nabla^x {\mathsf w}_x + 4\omega^x {\mathsf w}_x - \nabla^x{\mathsf w}_x - \nabla^\phi {\mathsf w}_\phi - 2 \omega^x{\mathsf w}_x - 2\omega^\phi {\mathsf w}_\phi 
- 4\omega^x u_x + 2 \omega^xu_x + 2\omega^\phi u_\phi = 0 \, ,
\end{split}
\end{equation}
\begin{equation}\label{eqn:xx}
2\omega^x u_x - 2\omega^\phi u_\phi = \nabla^x {\mathsf w}_x - \nabla^\phi {\mathsf w}_\phi + 2\omega^x {\mathsf w}_x - 2\omega^\phi {\mathsf w}_\phi \, .
\end{equation}
\noindent Component $(x\phi)$:
\begin{equation}
\nonumber
\nabla_x {\mathsf w}_\phi + \nabla_\phi {\mathsf w}_x + 2\omega_x {\mathsf w}_\phi + 2\omega_\phi {\mathsf w}_x - 2\omega_x u_\phi - 2\omega_\phi u_x = 0 \, ,
\end{equation}
\begin{equation}
\label{eqn:xphi}
2\omega_\phi u_x + 2\omega_x u_\phi = \nabla_x {\mathsf w}_\phi + \nabla_\phi {\mathsf w}_x + 2\omega_x {\mathsf w}_\phi + 2\omega_\phi {\mathsf w}_x \, .
\end{equation}
Finally we have (in matrix form)
\begin{equation}\label{eqn:eqnu}
\left[\begin{array}{c c}
2\omega^x & -2\omega^\phi \\
2\omega_\phi & 2\omega_x
\end{array}\right]\left[ \begin{array}{c}
u_x \\
u_\phi
\end{array}\right] = \left[ \begin{array}{c}
\nabla^x {\mathsf w}_x - \nabla^\phi {\mathsf w}_\phi + 2\omega^x {\mathsf w}_x - 2\omega^\phi {\mathsf w}_\phi \\
\nabla_x {\mathsf w}_\phi + \nabla_\phi {\mathsf w}_x + 2\omega_x {\mathsf w}_\phi + 2\omega_\phi {\mathsf w}_x
\end{array}\right] \, .
\end{equation}
Let us denote $A := \left[\begin{array}{c c}
2\omega^x & -2\omega^\phi \\
2\omega_\phi & 2\omega_x
\end{array}\right]$. Hence
\[A^{-1} = \frac{1}{4(\omega^x\omega_x + \omega^\phi\omega_\phi)}\left[\begin{array}{c c}
2\omega_x & 2\omega^\phi \\
-2\omega_\phi & 2\omega^x
\end{array}\right] = \frac{1}{2\|\omega\|^2}\left[\begin{array}{c c}
\omega_x & \omega^\phi \\
-\omega_\phi & \omega^x
\end{array}\right] \,.\]
Multiplying by $A^{-1}$ we get
\begin{equation}
\begin{split}\nonumber
u_x = & \frac{1}{2\|\omega\|^2}(\omega_x \nabla^x {\mathsf w}_x - \omega_x \nabla^\phi {\mathsf w}_\phi + 2\omega_x \omega^x {\mathsf w}_x - 2\omega_x \omega^\phi {\mathsf w}_\phi + \\
& + \omega^\phi \nabla_x {\mathsf w}_\phi + \omega^\phi \nabla_\phi {\mathsf w}_x + 2\omega_x \omega^\phi {\mathsf w}_\phi + 2 \omega_\phi \omega^\phi {\mathsf w}_x)
\end{split}
\end{equation}
or in simpler form
\begin{equation}\label{eqn:ux}
u_x = \frac{1}{2\|\omega\|^2}(\omega^x \nabla_x {\mathsf w}_x - \omega_x \nabla^\phi {\mathsf w}_\phi + \omega^\phi \nabla_x {\mathsf w}_\phi + \omega^\phi \nabla_\phi {\mathsf w}_x + 2 \|\omega\|^2 {\mathsf w}_x) \,.
\end{equation}
Similarly, component $\phi$:
\begin{equation}
\begin{split}\nonumber
u_\phi = & \frac{1}{2\|\omega\|^2}( \omega_\phi \nabla^\phi {\mathsf w}_\phi - \omega_\phi \nabla^x {\mathsf w}_x + 2\omega_\phi \omega^\phi {\mathsf w}_\phi - 2\omega_\phi \omega^x {\mathsf w}_x +\\
& + \omega^x \nabla_x {\mathsf w}_\phi + \omega^x \nabla_\phi {\mathsf w}_x + 2\omega^x \omega_x {\mathsf w}_\phi + 2\omega^x \omega_\phi {\mathsf w}_x)
\end{split}
\end{equation}
or
\begin{equation}\label{eqn:uphi}
u_\phi = \frac{1}{2\|\omega\|^2}( \omega_\phi \nabla^\phi {\mathsf w}_\phi - \omega_\phi \nabla^x {\mathsf w}_x + \omega^x \nabla_x {\mathsf w}_\phi + \omega^x \nabla_\phi {\mathsf w}_x + 2\|\omega\|^2 {\mathsf w}_\phi) \, .
\end{equation}
Equations (\ref{eqn:ux}) and (\ref{eqn:uphi}) we can rewrite in covariant form:
\begin{equation} \nonumber
u_A = \frac{1}{2\|\omega\|^2}(\omega^B \nabla_B {\mathsf w}_A + \omega^B \nabla_A {\mathsf w}_B - \omega_A \nabla^B {\mathsf w}_B + 2\|\omega\|^2 {\mathsf w}_A)
\end{equation}
Now, introducing $\Phi_A := \frac{1}{\|\omega\|^2}\omega_A$ we have
\begin{equation}\label{eqn:u1}
u_A = {\mathsf w}_A + \frac{1}{2}(\Phi^B \nabla_B {\mathsf w}_A + \Phi^B \nabla_A {\mathsf w}_B - \Phi_A \nabla^B {\mathsf w}_B) \, .
\end{equation}
Let us notice the following
\begin{equation}\label{eqn:rozw1}
\Phi^B \nabla_B {\mathsf w}_A = \nabla_B (\Phi^B {\mathsf w}_A) - {\mathsf w}_A \nabla_B \Phi^B = \nabla_B (\Phi^B {\mathsf w}_A) - {\mathsf w}_A
\end{equation}
\noindent(from \cite{JJcqg26} we know that $\nabla_B \Phi^B = 1$),
\begin{equation}\label{eqn:rozw2}
\Phi^B \nabla_A {\mathsf w}_B = \nabla_A (\Phi^B {\mathsf w}_B) - {\mathsf w}_B \nabla_A \Phi^B = \nabla_A (\Phi^B {\mathsf w}_B) - {\mathsf w}^B \nabla_A \Phi_B \, ,
\end{equation}
\begin{equation}\label{eqn:rozw3}
\Phi_A \nabla^B {\mathsf w}_B = \nabla^B (\Phi_A {\mathsf w}_B) - {\mathsf w}_B \nabla^B \Phi_A = \nabla_B (\Phi_A {\mathsf w}^B) - {\mathsf w}^B \nabla_B \Phi_A \, .
\end{equation}
The above equations (\ref{eqn:u1}), (\ref{eqn:rozw1}), (\ref{eqn:rozw2}) and (\ref{eqn:rozw3}) imply
\begin{equation}
\nonumber
u_A = {\mathsf w}_A + \frac{1}{2}\left[\nabla_B (\Phi^B {\mathsf w}_A - \Phi_A {\mathsf w}^B) + \nabla_A(\Phi^B {\mathsf w}_B) + {\mathsf w}^B (\nabla_B \Phi_A - \nabla_A \Phi_B) - {\mathsf w}_A \right] \, .
\end{equation}
From \cite{JJcqg26} we know that $\varepsilon^{AB}\nabla_B\Phi_A = 0$, hence $\nabla_B \Phi_A - \nabla_A \Phi_B = 0$, and we
obtain formula (\ref{eqn:u}):
$u_A = \frac{1}{2}\left[ {\mathsf w}_A + \nabla_A(\Phi^B {\mathsf w}_B) + \nabla_B (\Phi^B {\mathsf w}_A - \Phi_A {\mathsf w}^B) \right]$.

\subsection{Equations for ${\mathsf w}_A$}

\noindent The trace and curl of $u_A$ gives:
\begin{equation}
\label{eqn:trace2}
\nabla^A{\mathsf w}_A + 2\omega^A{\mathsf w}_A + \nabla^A u_A = 0 \, ,
\end{equation}
\begin{equation}
\label{eqn:rot2}
\varepsilon^{AB}\nabla_B u_A = 0
\end{equation}
\noindent Using formula
\begin{equation}
\label{eqn:u3}
u_A = \frac{1}{2}\left[ {\mathsf w}_A + \nabla_A(\Phi^B {\mathsf w}_B) + \nabla_B (\Phi^B {\mathsf w}_A - \Phi_A {\mathsf w}^B) \right]
\end{equation}
and equation (\ref{eqn:trace2}) we obtain
\begin{equation}\nonumber
\nabla^A {\mathsf w}_A + 2\omega^A {\mathsf w}_A + \frac{1}{2}\nabla^A{\mathsf w}_A + \frac{1}{2}\nabla^A \nabla_A (\Phi^B {\mathsf w}_B) + \frac{1}{2}\nabla^A\nabla_B (\Phi^B {\mathsf w}_A - \Phi_A {\mathsf w}^B) = 0\, ,
\end{equation}
\begin{equation}\nonumber
\Delta (\Phi^B {\mathsf w}_B) + 3\nabla^A {\mathsf w}_A + 4\omega^A {\mathsf w}_A + \nabla^A\nabla^B(\Phi_B {\mathsf w}_A - \Phi_A {\mathsf w}_B) = 0 \, .
\end{equation}
 Moreover, $\nabla^A\nabla^B(\Phi_B {\mathsf w}_A - \Phi_A {\mathsf w}_B) = 0$,
 because
\[\nabla_C\nabla_D t^{AB} - \nabla_D\nabla_C t^{AB} = R^A{_{ECD}}t^{EB} + R^B{_{ECD}}t^{AE} \, ,\]
where by $R^A{_{BCD}}$ we denote Riemann curvature tensor. We have
\[\nabla_A\nabla_B t^{AB} - \nabla_B \nabla_A t^{AB} = R_{EB}t^{EB} - R_{EA}t^{AE}\, , \]
where by $R_{AB}$ we denote Ricci tensor. The symmetry of Ricci
\[\nabla_A\nabla_B t^{AB} - \nabla_B \nabla_A t^{AB} = R_{EB}t^{EB} - R_{AE}t^{AE} = 0\]
implies
\begin{equation}\nonumber
\begin{split}
& \nabla^A\nabla^B(\Phi_B {\mathsf w}_A - \Phi_A {\mathsf w}_B) = 
\nabla^B\nabla^A(\Phi_B {\mathsf w}_A - \Phi_A {\mathsf w}_B) = \\
= & \nabla^A\nabla^B(\Phi_A {\mathsf w}_B - \Phi_B {\mathsf w}_A) = 
-\nabla^A\nabla^B(\Phi_B {\mathsf w}_A - \Phi_A {\mathsf w}_B) \, .
\end{split}
\end{equation}
Hence $\nabla^A\nabla^B(\Phi_B {\mathsf w}_A - \Phi_A {\mathsf w}_B) = 0$ and we obtain (\ref{laplw})
\begin{equation}\nonumber
\Delta (\Phi^B {\mathsf w}_B) + 3\nabla^A {\mathsf w}_A + 4\omega^A {\mathsf w}_A = 0 \, .
\end{equation}
Using formula (\ref{eqn:u3}) and equation (\ref{eqn:rot2}) we get
\[\varepsilon^{AB}\nabla_B {\mathsf w}_A + \varepsilon^{AB}\nabla_B\nabla^C(\Phi_C{\mathsf w}_A - \Phi_A{\mathsf w}_C) + \varepsilon^{AB}\nabla_B\nabla_A(\Phi^B{\mathsf w}_B) = 0\, .\]
Vanishing torsion gives $\varepsilon^{AB}\nabla_B\nabla_A(\Phi^B{\mathsf w}_B) = 0$, hence
\[\varepsilon^{AB}\nabla_B {\mathsf w}_A + \varepsilon^{AB}\nabla_B\nabla^C(\Phi_C{\mathsf w}_A - \Phi_A{\mathsf w}_C) = 0\, .\]
Moreover, $\Phi_C{\mathsf w}_A - \Phi_A{\mathsf w}_C = \varepsilon_{AC}\varepsilon^{DE}\Phi_D {\mathsf w}_E$
 implies
\[\varepsilon^{AB}\nabla_B {\mathsf w}_A + \varepsilon^{AB}\varepsilon_{AC}\nabla_B\nabla^C(\varepsilon^{DE}\Phi_D {\mathsf w}_E) \, .\]
Using identity $\varepsilon^{AB}\varepsilon_{AC} = -\delta^B{_C}$, we get
\[\varepsilon^{AB}\nabla_B {\mathsf w}_A + \nabla_C\nabla^C (\varepsilon^{AB}\Phi_B {\mathsf w}_A) = 0 \, , \]
and finally
we obtain (\ref{rotw})
\[\Delta(\varepsilon^{AB}\Phi_B {\mathsf w}_A) + \varepsilon^{AB}\nabla_B {\mathsf w}_A = 0 \, .\]

\section{Proof for large $k$}
Stability for the extremal Kerr leads to the following equation:
\begin{equation}
\label{eqn:vk}
 \partial_x a^2 \partial_x v_k - \frac{k^2}{a^2} v_k + Dv_k + \partial_x (a^2 B v_k) + ikCv_k = 0 \, ,
\end{equation}
where
\begin{itemize}
 \item $a^2 = 2\frac{1-x^2}{1+x^2}$
 \item $v_k : [-1,1] \to \mathbb{C}^2$ is the unknown function we are looking for, 
 \item $B = \frac{1}{1+x^2} \left[\begin{array}{cc}
                   x & -1 \\
                   3 & 3x
                  \end{array}\right]$,
 \item $C = \frac{1}{1+x^2} \left[\begin{array}{cc}
                   1 & x \\
                   -3x & 3
                  \end{array}\right]$,
 \item $D = \frac{4a^2}{1+x^2} \left[\begin{array}{cc}
                                       0 & 0 \\
                                       0 & 1
                                      \end{array}\right]$.
\end{itemize}
\begin{Theorem}\label{mainth}
Equation (\ref{eqn:vk}) has no solutions for $|k|>8$.
\end{Theorem}
\begin{proof}
For functions $f,g: [-1,1] \to \mathbb{C}^2$ let us define a standard scalar product: 
\[ (f|g) = \int\limits_{-1}^{1} \bar{f}^T g \, \rd x\]
Let us consider an operator $X: = a\frac{d}{dx}$ and its hermitian conjugate $X^* = -\frac{d}{dx}a$.
The eq. (\ref{eqn:vk}) takes the form:
\[X^* X v_k + \frac{k^2}{a^2} v_k + X^*(aBv_k) - ikCv_k - Dv_k = 0\]
The left-hand side we denote by $Lv_k$, where $L$ is a linear operator and $v_k \in \ker L$, i.e. $Lv_k = 0$.
For $(v_k|Lv_k)$ we have:
\[ 0 = (v_k|Lv_k) = \|Xv_k\|^2 + k^2 \left\|\frac{1}{a}v_k\right\| + (Xv_k|aBv_k) - ik\left(\frac{1}{a}v_k|aCv\right) - (v_k|Dv_k) \, . \]
 Introducing real numbers $x := \frac{\|Xv_k\|}{\|v_k\|}$, $y := \frac{\|\frac{1}{a}v_k\|}{\|v_k\|}$
 we obtain:
\[x^2 \|v_k\|^2 + k^2 y^2 \|v_k\|^2 = -(Xv_k|aBv_k) + ik\left(\frac{1}{a}v_k|aCv_k\right) + (v_k|Dv_k) \, , \]
and absolute value one can estimate as follows:
\[x^2 \|v_k\|^2 + k^2 y^2 \|v_k\|^2 \leq |(Xv_k|aBv_k)| + |k| \left|\left(\frac{1}{a}v_k|aCv_k\right)\right| + |(v_k|Dv_k)| \, .\]
From Cauchy-Schwarz inequality
\[x^2 \|v_k\|^2 + k^2 y^2 \|v_k\|^2 \leq x \|v_k\| \|aBv_k\| + |k|y\|v_k\|\|aCv_k\| + \|v_k\| \|Dv_k\| \, , \]
and from $\|Av\| \leq \|A\|\|v\|$ we get:
\[x^2 \|v_k\|^2 + k^2 y^2 \|v_k\|^2 \leq (x\|aB\| + |k|y\|aC\| + \|D\|)\|v_k\|^2 \, .\]
Hence
\[ x^2 + k^2 y^2 \leq x \|aB\| + |k| y \|aC\| + \|D\| \]
or in an equivalent form:
\[ \left(x-\frac{\|aB\|}{2}\right)^2 + \left(|k|y - \frac{\|aC\|}{2}\right)^2 \leq \|D\| + \frac{\|aB\|^2 + \|aC\|^2}{4} \,.\]
Positivity of $(x-\frac{\|aB\|}{2})^2$ gives
\[\left(|k|y - \frac{\|aC\|}{2}\right)^2 \leq \|D\| + \frac{\|aB\|^2 + \|aC\|^2}{4}\]
and
\[ |k| \leq \frac{\|aC\| + \sqrt{4\|D\| + \|aB\|^2 + \|aC\|^2} }{2y} \, .\]
Definition of $y = \frac{\|\frac{1}{a}v_k\|}{\|v_k\|}$ and $a \leq \sqrt{2}$ gives $y \geq \frac{1}{\sqrt{2}}$, hence
\[ |k| \leq \frac{\|aC\| + \sqrt{4\|D\| + \|aB\|^2 + \|aC\|^2} }{\sqrt{2}} \, . \]
A simple computation gives $\|D\| = 8$, $\|aB\| = \sqrt{6}$, $\|aC\| = 3\sqrt{2}$. Finally
\[ |k| \leq \frac{3\sqrt{2} + \sqrt{32 + 6 + 18}}{\sqrt{2}} = 3 + \sqrt{28} \approx 8.29 \, ,\]
but $k$ is integer hence $|k| \leq 8$.
\end{proof}

\section{Conformal vector field for extremal Kerr}\label{cvf}

We are looking for a vector field $X$ in the following form:
\[ X = A(x) \cos\phi \partial_x + B(x) \sin\phi \partial_\phi \, .\]
In coordinates $(x,\phi)$ the metric tensor
$\displaystyle (g_{AB}) = \left(\begin{array}{cc}
m^2 \frac{1+x^2}{1-x^2} & 0 \\
0 & 4m^2 \frac{1-x^2}{1+x^2}
\end{array}\right)$,
hence
\[X_x = Am^2\frac{1+x^2}{1-x^2}\cos\phi \, ,\quad X_\phi = 4Bm^2\frac{1-x^2}{1+x^2}\sin\phi \, .\]

The CVF equation
\[ \nabla_A X_B + \nabla_B X_A = \nabla_C X^C g_{AB} \]
applied to our field $X$ reduces to
\[\nabla_x X_x = \left( \frac{1+x^2}{1-x^2}A' + \frac{2x}{(1-x^2)^2}A \right) m^2\cos\phi \, ,\]
\[\nabla_\phi X_\phi = 4m^2 \left(\frac{1-x}{1+x^2}B - \frac{2x}{(1+x^2)^2}A\right) \cos\phi \, , \]
\[\nabla_\phi X_x + \nabla_x X_\phi = \left(4\frac{1-x^2}{1+x^2}B' - \frac{1+x^2}{1-x^2}A\right) m^2 \sin\phi \, , \]
\[ \nabla_C X^C = (A'+B)\cos\phi \, ,\]
where $A' := \frac{dA}{dx}$.
They can be written in an equivalent form:
\[\left( \frac{1+x^2}{1-x^2}A' + \frac{2x}{(1-x^2)^2}A \right) m^2\cos\phi = \frac{1}{2}m^2\frac{1+x^2}{1-x^2}(A'+B)\cos\phi \, , \]
\[4m^2 \left(\frac{1-x}{1+x^2}B - \frac{2x}{(1+x^2)^2}A\right) \cos\phi = 2m^2\frac{1-x^2}{1+x^2}(A'+B)\cos\phi \, , \]
\[4\frac{1-x^2}{1+x^2}B' - \frac{1+x^2}{1-x^2}A = 0 \, , \]
and finally we obtain system of ODE's:
\[B = A' + \frac{4x}{1-x^4}A \, ,\]
\[4B'\frac{1-x^2}{1+x^2} = A\frac{1+x^2}{1-x^2} \, ,\]
which leads to the second order ODE for the function $B$:
\[\left[\partial_x^2 - \frac{4x}{1-x^4}\partial_x - \frac{1}{4}\left(\frac{1+x^2}{1-x^2}\right)^2\right] B(x) = 0\]
and
$A(x) = 4\left(\frac{1-x^2}{1+x^2}\right)^2 B'(x)$.
We get the following solution:
\[B(x) = C_1 \cosh\left[\frac{1}{2}\left(x+\log\frac{1-x}{1+x}\right)\right] + C_2 \sinh\left[\frac{1}{2}\left(x+\log\frac{1-x}{1+x}\right)\right] \, . \]
If we assume that the field $X$ vanishes at one ``pole'' ($x=\pm 1$) we obtain the relation for constants $C_i$:
$C_2 = \pm C_1$. For $C_2 = -C_1$ we have:
\[B(x) = C\sqrt{\frac{1+x}{1-x}}e^{-\frac{1}{2}x} \, , \quad
 A(x) = 2C\sqrt{\frac{1+x}{1-x}}\frac{1-x^2}{1+x^2}e^{-\frac{1}{2}x}\, .\]

Finally the CVF $X$ takes the form:
\[X = C\sqrt{\frac{1+x}{1-x}}e^{-\frac{1}{2}x}\left(2\frac{1-x^2}{1+x^2}\cos\phi\partial_x +
\sin\phi \partial_\phi\right) \, .\]

\section*{Acknowledgements}
This research was supported by Polish Ministry of Science and Higher
Education grant Nr N N201 372736.

\end{document}